\documentclass{svproc}
\usepackage{url}

 \newtheorem{thm}{Theorem}[section]

 \newtheorem{prop}[thm]{Proposition}

\usepackage[utf8]{inputenc}

\usepackage[ruled,vlined,linesnumbered]{algorithm2e}
\SetKwProg{Fn}{Function}{}{}

\usepackage{enumitem}
\usepackage{float}

\usepackage{color}
\usepackage{xcolor}
\definecolor{deepblue}{rgb}{0,0,0.5}
\definecolor{deepred}{rgb}{0.6,0,0}
\definecolor{deepgreen}{rgb}{0,0.5,0}
\definecolor{verydarkGray}{RGB}{30,30,30}
\definecolor{darkGray}{RGB}{80,80,80}
\definecolor{midGray}{RGB}{160,160,160}  
\definecolor{midBlue}{RGB}{120,120,160}  
\definecolor{verylightgray}{RGB}{210,210,210}
\definecolor{verylightred}{rgb}{0.8,0.1,0.1}
\definecolor{gray}{rgb}{0.4,0.4,0.4}
\definecolor{darkblue}{rgb}{0.0,0.0,0.6}
\definecolor{cyan}{rgb}{0.0,0.6,0.6}
\definecolor{forestGreen}{RGB}{0, 153, 76}
\usepackage{colortbl}
\definecolor{folderbg}{RGB}{124,166,198}
\definecolor{folderborder}{RGB}{110,144,169}

\usepackage{mathptmx}
\usepackage{times}
\usepackage{amsmath}
\usepackage{fixmath}
\usepackage{mathpazo} 
\usepackage{amssymb} 

\usepackage{tikz}
\usepackage{tikz-qtree,tikz-qtree-compat} 
\usepackage{tikz-cd} 
\usetikzlibrary{arrows, matrix}
\usetikzlibrary{trees}
\usetikzlibrary{shadings} 
\usetikzlibrary{intersections}
\usetikzlibrary{positioning}
\usepackage{caption}
\usepackage{multirow}
\usetikzlibrary{automata}
\usepackage{forest}
\def\Size{4pt}
\tikzset{
  folder/.pic={
    \filldraw[draw=folderborder,top color=folderbg!50,bottom color=folderbg]
      (-1.05*\Size,0.2\Size+5pt) rectangle ++(.75*\Size,-0.2\Size-5pt);  
    \filldraw[draw=folderborder,top color=folderbg!50,bottom color=folderbg]
      (-1.15*\Size,-\Size) rectangle (1.15*\Size,\Size);
  }
}

\title{Efficient representation and manipulation of quadratic surfaces using Geometric Algebras}
\author{St\'ephane Breuils\inst{1}, Vincent Nozick\inst{1}, Laurent Fuchs\inst{2} and Akihiro Sugimoto\inst{3}}
\institute{
Laboratoire d'Informatique Gaspard-Monge, Equipe A3SI,\\ UMR 8049, Universit\'e Paris-Est Marne-la-Vall\'ee, France\\
France\\
\email{firstName.name@u-pem.fr}
\and
XLIM-ASALI,~UMR 7252, \\ Universit\'e de Poitiers, Poitiers, France\\
\email{Laurent.Fuchs@univ-poitiers.fr}
\and
National Institute of Informatics, \\Tokyo 101-8430, Japan,\\
\email{sugimoto@nii.ac.jp}
}

\date{March 2019}

\begin{document}
\mainmatter
\titlerunning{GA \& quadratic surfaces }
\maketitle

\begin{abstract}
Quadratic surfaces gain more and more attention among the Geometric Algebra community and some frameworks were proposed in order to represent, transform, and intersect these quadratic surfaces. As far as the authors know, none of these frameworks support all the operations required to completely handle these surfaces. Some frameworks do not allow the construction of quadratic surfaces from control points when others do not allow to transform these quadratic surfaces. However, if we consider all the frameworks together, then all the required operations over quadratic are covered. This paper presents a unification of these frameworks that enables to represent any quadratic surfaces either using control points or from the coefficients of its implicit form. The proposed approach also allows to transform any quadratic surfaces and to compute their intersection and to easily extract some geometric properties.
\end{abstract}

\section{Introduction}
Geometric algebra provides convenient and intuitive tools to represent, transform, and intersect geometric objects. Deeply explored by physicists, it has been used in quantum mechanics and electromagnetism~\cite{Hestenes1990,gregory2017elastic} as well as in classical mechanics~\cite{hestenes2012new}. Geometric algebra has also found some interesting applications in geographic data manipulations~\cite{Luo2017,Zhu2018}. Among them, geometric algebra is used within the computer graphics community. More precisely, it is used not only in basis geometric primitive manipulations~\cite{vince2008geometric} but also in complex illumination processes as in~\cite{papaefthymioureal} where spherical harmonics are substituted by geometric algebra entities. Finally, in data and image analysis, we can find the usefulness of geometric algebra in mathematical morphology~\cite{dorst1993analytical} and in neural networking~\cite{buchholz2007optimal,hitzer2013geometric}.
In the geometric algebra community, quadratic surfaces gain more and more attention, and some frameworks have been proposed in order to represent, transform, and intersect these quadratic surfaces. 

There exist three main approaches to deal with quadratic surfaces. The first one, introduced in~\cite{easter20162CGA}, is called double conformal geometric algebra of~$\mathbb{G}_{6,2}$. It is capable of representing quadratics surfaces from the coefficients of their implicit form.  The second one is double perspective geometric algebra of $\mathbb{G}_{4,4}$ whose definition was firstly introduced in~\cite{Goldman2015}. It has been further developed in~\cite{Goldman2017}. This approach is based on a duplication of $\mathbb{R}^3$ and also represents quadratic surfaces from the coefficients of their implicit form, as bivectors. 
To our best knowledge, it can not construct general quadrics from control points.
The third one was introduced in~\cite{Breuils2018} and is denoted as quadric conformal geometric algebra (QCGA). 
QCGA allows us to define general quadrics from $9$ control points, and to represent the objects in low dimensional subspaces of the algebra. QCGA is capable of constructing quadratic surfaces either using control points or implicit equations as $1$-vector. As far as we know, QCGA does not yet allow all geometric transformations of quadrics surfaces.  


In order to enhance usefulness of Geometric Algebra for geometry and computer graphics community, a new framework that allows to represent and manipulate quadratic surfaces has to be developed.  This is the main purpose of this paper.

All the examples and computations are based upon the efficient geometric algebra library generator Garamon~\cite{breuilsagaramon}. The code of this library generator is available online\footnote{git clone https://git.renater.fr/garamon.git}.

\subsection{Contributions}
We propose a new approach that unifies the three above framework. Using this approach, we show that it is possible to represent quadratic surfaces either using control points or its implicit coefficients but also transform these quadratic surfaces using versors. Computation of the tangent planes to a quadratic surfaces is also defined. We show that it is also possible to intersect quadratic surfaces in this new framework. Finally, this approach also enables to define also some quartic surfaces. 

\section{Notations:}
Following the state-of-the-art usages of~\cite{LeoBook} and~\cite{perwassGeometric}, upper-case bold letters denote blades (blade~$\mathbf{A}$) whose grade is higher than $1$. Multivectors and $k$-vectors are denoted with upper-case non-bold letters (multivector~$A$). Lower-case bold letters refer to vectors and lower-case non-bold to multivector coordinates. The $k$-grade part of a multivector $A$ is denoted by $\langle A \rangle_{k}$. Finally, The vector space dimension is denoted by $2^d$, where~$d$ is the number of basis blades $\mathbf{e}_{i}$ of grade~$1$.

\section{Computational efficiency of models of Quadric surfaces with Geometric Algebras}
\subsection{Complexity estimation model}
\label{sec:quadricComparisonModel}
~\\In order to compare the operations for these different frameworks, we need a computational model. This requires to be able to determine  the complexity of operations in each framework. Using the complexity of the operations explained in the last part of this manuscript is well suited when one wants to compare different methods to compute the products. Instead, our complexity estimation is made through two simplifying assumptions. 

~\\First, let us consider the outer product between one homogeneous multivector whose number of components is $u$ and another homogeneous multivector whose number of components is $v$, $u,v\in \mathbb{N}$. We assume that an upper bound to the number of required products is at most $uv$ products, as shown in the definition of the outer product.

~\\The equation of the inner product is our base for the second assumption. Furthermore, we need to use this formula for inner products between $1-$vector and $2-$vector as well as inner product between two $1-$vector. The first multivector has $u$ non-zero components, and the second has $v$ non-zero components. Then the inner product between two $1-$vectors will result in $uv$ products. Whereas, the inner product between $1-$vectors and $2$-vectors requires two inner products for each pair of components of the two multivectors. Thus, this inner product requires at most $2uv$ products.   
These models will be used for determining the complexity of each operation. Let us now explain in more details these models, we start with DCGA.

\subsection{Models of quadratic surfaces with Geometric Algebra}
This section presents the main Geometric Algebra models to represent and manipulate quadratic surfaces. Furthermore, we aim at determining the most efficient Geometric Algebra model for each of the main operations required in computer graphics. The purpose is to propose a Geometric Algebra model that allows to efficiently handle quadrics.
~\\The considered operations over the surfaces are: 
\begin{itemize}
\item checking if a point lies in a quadratic surface,
\item intersecting quadratic surface and line,
\item computing the normal vector (and tangent plane) of a surface at a given point.
\end{itemize}
One of the applications of such operations is to compute precise visualisations using ray-tracer~\cite{glassner1989introduction}. 

~\\A first framework to handle quadratic surfaces was introduced thanks to the pioneering work of Zamora~\cite{Zamora2014}. This framework constructs a quadratic surface from control points in $\mathbb{G}_{6,3}$. In this model, an axis-aligned quadratic surface $\mathbf{Q}$ can be defined as:
\begin{equation}
\mathbf{Q} = \mathbf{x}_1 \wedge \mathbf{x}_2 \wedge \mathbf{x}_3 \wedge \mathbf{x}_4 \wedge \mathbf{x}_5 \wedge \mathbf{x}_6
\end{equation}  
The major drawback of this model is that it supports only axis-aligned quadratic surfaces. Due to this fact, we will not further detail this model.

There exist three main Geometric Algebra frameworks to manipulate general quadratic surfaces. First, DCGA (Double Conformal Geometric Algebra) with $\mathbb{G}_{8,2}$ defined by Easter and Hitzer~\cite{easter20162CGA}. Second, a framework of~$\mathbb{G}_{4,4}$ as firstly introduced by Parkin~\cite{Parkin} and developed further by Du et~al.~\cite{Goldman2017}. The last one is our contribution, introduced in~\cite{Breuils2018} and is a model of $\mathbb{G}_{9,6}$.

\subsubsection{DCGA of $\mathbb{G}_{8,2}$}
DCGA was presented by Hitzer and Easter~\cite{easter20162CGA} and aims at having entities representing both quartic surfaces and quadratic surfaces.
~\\\textbf{Basis and metric}~\\
In more details, the DCGA $\mathbb{G}_{8,2}$ is defined over a $10$-dimensional vector space. The base vectors of the space are basically divided into two groups: $ \{ \mathbf{e}_{o1},\mathbf{e}_{1}, \mathbf{e}_{2}, \mathbf{e}_{3}, \mathbf{e}_{\infty 1}\} $, corresponding to the CGA vectors defined in Chapter 1, and a copy of this basis $ \{ \mathbf{e}_{o2},\mathbf{e}_{4}, \mathbf{e}_{5}, \mathbf{e}_{6}, \mathbf{e}_{\infty 2}\} $. The inner products between them are defined in Table~\ref{table:metricDCGA}. Note that we highlighted in grey the two models of CGA that are included in DCGA.

\begin{table}[ht]
\caption{Inner product between DCGA basis vectors.}
$$
\begin{array}{c|rrrrrrrrrr}
      & \mathbf{e}_{o1} & \mathbf{e}_1 & \mathbf{e}_2 & \mathbf{e}_3 & \mathbf{e}_{\infty 1} & \mathbf{e}_{o2}& \mathbf{e}_{4} & \mathbf{e}_{5} & \mathbf{e}_{6} & \mathbf{e}_{\infty 2} \\ \hline 
 
\mathbf{e}_{o1} & \cellcolor{verylightgray} \textcolor{gray}{0} & \cellcolor{verylightgray}\textcolor{gray}{0} & \cellcolor{verylightgray}\textcolor{gray}{0} & \cellcolor{verylightgray}\textcolor{gray}{0} & \cellcolor{verylightgray} -1 & \textcolor{gray}{0} & \textcolor{gray}{0} & \textcolor{gray}{0} & \textcolor{gray}{0} & \textcolor{gray}{0} \\

\mathbf{e}_1 & \cellcolor{verylightgray}\textcolor{gray}{0} & \cellcolor{verylightgray}1 & \cellcolor{verylightgray}\textcolor{gray}{0} & \cellcolor{verylightgray}\textcolor{gray}{0} & \cellcolor{verylightgray}\textcolor{gray}{0} & \textcolor{gray}{0} & \textcolor{gray}{0} & \textcolor{gray}{0} & \textcolor{gray}{0} & \textcolor{gray}{0} \\

\mathbf{e}_2 & \cellcolor{verylightgray}\textcolor{gray}{0} & \cellcolor{verylightgray}\textcolor{gray}{0} & \cellcolor{verylightgray}1 & \cellcolor{verylightgray}\textcolor{gray}{0} & \cellcolor{verylightgray}\textcolor{gray}{0} & \textcolor{gray}{0} & \textcolor{gray}{0} & \textcolor{gray}{0} & \textcolor{gray}{0} & \textcolor{gray}{0} \\

\mathbf{e}_3 & \cellcolor{verylightgray}\textcolor{gray}{0} & \cellcolor{verylightgray}\textcolor{gray}{0} & \cellcolor{verylightgray}\textcolor{gray}{0} & \cellcolor{verylightgray}1 & \cellcolor{verylightgray}\textcolor{gray}{0} & \textcolor{gray}{0} & \textcolor{gray}{0} & \textcolor{gray}{0} & \textcolor{gray}{0} & \textcolor{gray}{0} \\

\mathbf{e}_{\infty 1}  & \cellcolor{verylightgray} -1 & \cellcolor{verylightgray}\textcolor{gray}{0} & \cellcolor{verylightgray}\textcolor{gray}{0} & \cellcolor{verylightgray}\textcolor{gray}{0} & \cellcolor{verylightgray}\textcolor{gray}{0} & \textcolor{gray}{0} & \textcolor{gray}{0} & \textcolor{gray}{0} & \textcolor{gray}{0} & \textcolor{gray}{0} \\
 
\mathbf{e}_{o2} & \textcolor{gray}{0} & \textcolor{gray}{0} & \textcolor{gray}{0} & \textcolor{gray}{0} & \textcolor{gray}{0} & \cellcolor{verylightgray}\textcolor{gray}{0} & \cellcolor{verylightgray}\textcolor{gray}{0} & \cellcolor{verylightgray}\textcolor{gray}{0} & \cellcolor{verylightgray}\textcolor{gray}{0} & \cellcolor{verylightgray} -1 \\
 
\mathbf{e}_4 & \textcolor{gray}{0} & \textcolor{gray}{0} & \textcolor{gray}{0} & \textcolor{gray}{0} & \textcolor{gray}{0} & \cellcolor{verylightgray}\textcolor{gray}{0} & \cellcolor{verylightgray}1 & \cellcolor{verylightgray}\textcolor{gray}{0} & \cellcolor{verylightgray}\textcolor{gray}{0} & \cellcolor{verylightgray}\textcolor{gray}{0} \\
 
\mathbf{e}_5 & \textcolor{gray}{0} & \textcolor{gray}{0} & \textcolor{gray}{0} & \textcolor{gray}{0} & \textcolor{gray}{0} & \cellcolor{verylightgray}\textcolor{gray}{0} & \cellcolor{verylightgray}\textcolor{gray}{0} & \cellcolor{verylightgray}1 & \cellcolor{verylightgray}\textcolor{gray}{0} & \cellcolor{verylightgray}\textcolor{gray}{0} \\
 
\mathbf{e}_6 & \textcolor{gray}{0} & \textcolor{gray}{0} & \textcolor{gray}{0} & \textcolor{gray}{0} & \textcolor{gray}{0} & \cellcolor{verylightgray}\textcolor{gray}{0} & \cellcolor{verylightgray}\textcolor{gray}{0} & \cellcolor{verylightgray}\textcolor{gray}{0} & \cellcolor{verylightgray}1 & \cellcolor{verylightgray}\textcolor{gray}{0} \\
 
\mathbf{e}_{\infty2} & \textcolor{gray}{0} & \textcolor{gray}{0} & \textcolor{gray}{0} & \textcolor{gray}{0} & \textcolor{gray}{0} & \cellcolor{verylightgray}-1 & \cellcolor{verylightgray}\textcolor{gray}{0} & \cellcolor{verylightgray}\textcolor{gray}{0} & \cellcolor{verylightgray}\textcolor{gray}{0} & \cellcolor{verylightgray}\textcolor{gray}{0} \\
\end{array} 
$$
\label{table:metricDCGA}
\end{table}

~\\ \textbf{Point of DCGA}~\\
A point of DCGA whose Euclidean coordinates are $(x,y,z)$ is defined as the outer product of two CGA points with coordinates $(x,y,z)$. By defining this two points $\mathbf{x}_1$ and $\mathbf{x}_2$ as:
\begin{equation}
\begin{array}{cc}
\mathbf{x}_1 &= \mathbf{e}_{o1} + x \mathbf{e}_{1} + y \mathbf{e}_{2} + z \mathbf{e}_{3} + \frac{1}{2} (x^2 + y^2 + z^2)\mathbf{e}_{\infty1} \\
\mathbf{x}_2 &= \mathbf{e}_{o2} + x \mathbf{e}_{4} + y \mathbf{e}_{5} + z \mathbf{e}_{6} + \frac{1}{2} (x^2 + y^2 + z^2)\mathbf{e}_{\infty2} \\
\end{array}
\end{equation}
The embedding of a DCGA point is defined as follows:
\begin{equation}
\mathbf{X} = \mathbf{x}_1 \wedge \mathbf{x}_2
\end{equation}
The development of this operation results in:
\begin{equation}
\begin{array}{cll}
\mathbf{X} &=&  \color{red}x^2 \mathbf{e}_{1} \wedge \mathbf{e}_{4} + y^2 \mathbf{e}_{2} \wedge \mathbf{e}_{5} + z^2 \mathbf{e}_{3} \wedge \mathbf{e}_{6}
+ xy (\mathbf{e}_{1} \wedge \mathbf{e}_{5} + \mathbf{e}_{2} \wedge \mathbf{e}_{4} ) \\
&&\color{red} + xz (\mathbf{e}_{1} \wedge \mathbf{e}_{6} + \mathbf{e}_{3} \wedge \mathbf{e}_{4}) + yz (\mathbf{e}_{2} \wedge \mathbf{e}_{6} + \mathbf{e}_{3} \wedge \mathbf{e}_{5} ) \\
&& \color{red} +  x (\mathbf{e}_{o1} \wedge \mathbf{e}_{4}+\mathbf{e}_{1} \wedge \mathbf{e}_{o2}) + y (\mathbf{e}_{o1} \wedge \mathbf{e}_{5}+\mathbf{e}_{2} \wedge \mathbf{e}_{o2})\\ &&{\color{red}+ z (\mathbf{e}_{o1} \wedge \mathbf{e}_{6} +\mathbf{e}_{3} \wedge \mathbf{e}_{o2}) + \mathbf{e}_{o1} \wedge \mathbf{e}_{o2}}\\
&& +  \frac{1}{2} (x^2 + y^2 + z^2) (\mathbf{e}_{o1} \wedge \mathbf{e}_{\infty2} + \mathbf{e}_{\infty1} \wedge \mathbf{e}_{o2}) \\
&&+\frac{1}{2} x (x^2 + y^2 + z^2) (\mathbf{e}_{1} \wedge \mathbf{e}_{\infty2} + \mathbf{e}_{\infty1} \wedge \mathbf{e}_{4}) \\
&&+\frac{1}{2} y (x^2 + y^2 + z^2) (\mathbf{e}_{2} \wedge \mathbf{e}_{\infty2} + \mathbf{e}_{\infty1} \wedge \mathbf{e}_{5})\\
&&+\frac{1}{2} z (x^2 + y^2 + z^2) (\mathbf{e}_{3} \wedge \mathbf{e}_{\infty2} + \mathbf{e}_{\infty1} \wedge \mathbf{e}_{6})\\
&& +\frac{1}{4}  (x^2 + y^2 + z^2)^2 \mathbf{e}_{\infty1} \wedge \mathbf{e}_{\infty2} 
\end{array}
\end{equation}
This high number of components is due to the fact that the representation of a point of DCGA was designed to non only define quadratic surfaces but also quartic surfaces. To illustrate this, we highlight the components that contribute to construct quadrics in {\color{red}red}. The other components are dedicated to the construction of quartics.
 
~\\ \textbf{Quadratic surfaces}~\\
A general quadratic surface merely consists in defining some operators that extract the components of $\mathbf{x}$. For a general quadric defined as:
\begin{equation}
\mathrm{a} x^2+\mathrm{b} y^2 +\mathrm{c}z^2+\mathrm{d}xy+\mathrm{e}yz+\mathrm{f}zx+\mathrm{g}x+\mathrm{h} y  + \mathrm{i}z + \mathrm{j} =0
\end{equation}
This means that $4$ operators are defined for the quadratic part:
\begin{equation}
\begin{array}{c@{}lc@{}l}
\mathbf{T}_{x^{2}} &= \mathbf{e}_4 \wedge \mathbf{e}_1 & \mathbf{T}_{y^{2}} &= \mathbf{e}_5 \wedge \mathbf{e}_2\\
\mathbf{T}_{z^{2}} &= \mathbf{e}_6 \wedge \mathbf{e}_3 & \mathbf{T}_{1} &= -\mathbf{e}_{\infty1} \wedge \mathbf{e}_{\infty2} \\
\end{array}
\end{equation}
along with the $3$ operators for the linear part:
\begin{equation}
\begin{array}{cl}
\mathbf{T}_{x} &= \frac{1}{2}  \Big( \mathbf{e}_1 \wedge \mathbf{e}_{\infty 2} + \mathbf{e}_{\infty 1} \wedge \mathbf{e}_4 \Big)\\
 \mathbf{T}_{y} &= \frac{1}{2}  \Big( \mathbf{e}_2 \wedge \mathbf{e}_{\infty 2} + \mathbf{e}_{\infty 1} \wedge \mathbf{e}_5 \Big) \\
\mathbf{T}_{z} &= \frac{1}{2}  \Big( \mathbf{e}_3 \wedge \mathbf{e}_{\infty 2} + \mathbf{e}_{\infty 1} \wedge \mathbf{e}_6 \Big) \\ 
\end{array}
\end{equation}
and $3$ operators for the cross terms:
\begin{equation}
\begin{array}{cl}
\mathbf{T}_{xy} &= \frac{1}{2} \Big(\mathbf{e}_5 \wedge \mathbf{e}_1 + \mathbf{e}_4 \wedge \mathbf{e}_2 \Big)\\
\mathbf{T}_{xz} &= \frac{1}{2} \Big(\mathbf{e}_6 \wedge \mathbf{e}_1 + \mathbf{e}_4 \wedge \mathbf{e}_3 \Big)\\
\mathbf{T}_{yz} &= \frac{1}{2} \Big(\mathbf{e}_5 \wedge \mathbf{e}_3 + \mathbf{e}_6 \wedge \mathbf{e}_2\Big)\\
\end{array}
\end{equation}
Then, for example,
\begin{equation}
\mathbf{T}_{yz} \cdot \mathbf{x} = yz
\end{equation}
Finally a general quadratic surface is defined as the bivector $\mathbf{Q}_{DCGA}$ with the following formula:
\begin{equation}
\mathbf{Q}_{DCGA} = \mathrm{a} \mathbf{T}_{x^2} +\mathrm{b} \mathbf{T}_{y^2} +\mathrm{c}\mathbf{T}_{z^2}+\mathrm{d}\mathbf{T}_{xy} + \mathrm{e}\mathbf{T}_{yz} +\mathrm{f}\mathbf{T}_{xz}+\mathrm{g}\mathbf{T}_{x} +\mathrm{h} \mathbf{T}_{y}  + \mathrm{i}\mathbf{T}_{z} + \mathrm{j}\mathbf{T}_{1}
\end{equation}
Finally we check that a point $\mathbf{x}$ is in a quadratic surface if and only if:
\begin{equation}
\mathbf{Q}_{DCGA} \cdot \mathbf{x} = 0
\end{equation}
DCGA not only supports the definition of general quadrics but also some quartic surfaces like Torus, cyclides (Dupin cyclides...).

~\\\textbf{Plane tangent to a quadratic surface}~\\
The tangent plane was defined using differential operators in DCGA. Let us consider a point $\mathbf{x}$ whose Euclidean coordinates are ($x,y,z$) and a DCGA quadric $\mathbf{Q}_{DCGA}$ defined as:
\begin{equation}
\mathbf{Q}_{DCGA} = \mathrm{a} \mathbf{T}_{x^2} +\mathrm{b} \mathbf{T}_{y^2} +\mathrm{c}\mathbf{T}_{z^2}+\mathrm{d}\mathbf{T}_{xy} + \mathrm{e}\mathbf{T}_{yz} +\mathrm{f}\mathbf{T}_{xz}+\mathrm{g}\mathbf{T}_{x} +\mathrm{h} \mathbf{T}_{y}  + \mathrm{i}\mathbf{T}_{z} + \mathrm{j}\mathbf{T}_{1}
\end{equation}
The differential operators along the axis are defined as:
\begin{equation}
\begin{array}{c@{}c@{}}
\mathbf{D}_{x} &= (\mathbf{e}_1 \wedge \mathbf{e}_{\infty 1} + \mathbf{e}_{4} \wedge \mathbf{e}_{\infty 2})\\ 
\mathbf{D}_{y} &= (\mathbf{e}_2 \wedge \mathbf{e}_{\infty 1} + \mathbf{e}_{5} \wedge \mathbf{e}_{\infty 2})\\ 
\mathbf{D}_{z} &= (\mathbf{e}_3 \wedge \mathbf{e}_{\infty 1} + \mathbf{e}_{6} \wedge \mathbf{e}_{\infty 2})\\ 
\end{array}
\end{equation}
Then, using the commutator product, noted as $\times$, the following properties hold:
\begin{equation}
\begin{array}{c@{}c@{}}
\mathbf{D}_{x} \times \mathbf{q}_{DCGA}  &= 2 \mathrm{a} \mathbf{T}_{x} + \mathrm{d} \mathbf{T}_{y} + \mathrm{e} \mathbf{T}_{z} + \mathrm{g} \mathbf{T}_{1}\\ 
\mathbf{D}_{y} \times \mathbf{q}_{DCGA}  &= 2 \mathrm{b} \mathbf{T}_{y} + \mathrm{d} \mathbf{T}_{x} + \mathrm{f} \mathbf{T}_{z} + \mathrm{h} \mathbf{T}_{1}\\ 
\mathbf{D}_{z} \times \mathbf{q}_{DCGA}  &= 2 \mathrm{c} \mathbf{T}_{z} + \mathrm{e} \mathbf{T}_{x} + \mathrm{f} \mathbf{T}_{y} + \mathrm{i} \mathbf{T}_{1}\\ 
\end{array}
\label{eq:differentialOperators}
\end{equation}
This latter formula defines the normal vector to the quadric surface at any point of the surface. It is computed using the normal vector $\mathbf{n}_1$ defined in the first copy CGA and $\mathbf{n}_2$ along the second copy of the CGA basis vectors. $\mathbf{n}_1$ at the considered point can be defined as follows:
\begin{equation}
\mathbf{n}_{1} = ((\mathbf{D}_{x} \times \mathbf{q}_{DCGA}) \cdot \mathbf{x}) \mathbf{e}_1 + ((\mathbf{D}_{y} \times \mathbf{q}_{DCGA}) \cdot \mathbf{x}) \mathbf{e}_2 + ((\mathbf{D}_{z} \times \mathbf{q}_{DCGA}) \cdot \mathbf{x}) \mathbf{e}_3 
\end{equation}
Similarly, the normal vector $\mathbf{n}_2$ is:
\begin{equation}
\mathbf{n}_{2} = ((\mathbf{D}_{x} \times \mathbf{q}_{DCGA}) \cdot \mathbf{x}) \mathbf{e}_4 + ((\mathbf{D}_{y} \times \mathbf{q}_{DCGA}) \cdot \mathbf{x}) \mathbf{e}_5 + ((\mathbf{D}_{z} \times \mathbf{q}_{DCGA}) \cdot \mathbf{x}) \mathbf{e}_6 
\end{equation} 
Now, the definition of the plane from normal vector is:
\begin{equation}
\mathbold{\Pi} = (\mathbf{n}_{1} + \mathrm{d} \mathbf{e}_{\infty1} ) \wedge (\mathbf{n}_{2} + \mathrm{d} \mathbf{e}_{\infty2} )
\end{equation}
where $\mathrm{d}$ represents the orthogonal distance between the plane and the origin.
Finally, the computation of the orthogonal distance can be simply performed as follows:
\begin{equation}
   \mathrm{d} = \mathbf{n}_1 \cdot \mathbf{x}_1
\end{equation}
Where $\mathbf{x}_1$ is the point used to form the DCGA point $\mathbf{x}$.

~\\\textbf{Quadric-line intersection}~\\
DCGA also supports the construction of the intersection $\mathbf{p}_p$ of a quadric $\mathbf{Q}_{DCGA}$ and a line $\mathbf{L}$. A line in DCGA can be defined as:
\begin{equation}
\mathbf{L} = \mathbf{l}_1 \wedge \mathbf{l}_2
\end{equation}
Where these two entities $\mathbf{l}_1$ and $\mathbf{l}_2$ can be expressed using the direction of the unit vector in CGA1 $\mathbf{\mathbf{d}_1}$ and CGA2 $\mathbf{d}_1$ and a point of this line expressed in CGA1 $\mathbf{x}_1$ and CGA2 $\mathbf{x}_2$ as:
\begin{equation}
\mathbf{l}_1 = \mathbf{d}_1 \mathbf{I}_{\epsilon}^{-1} - (\mathbf{x}_1\cdot (\mathbf{d}_1 \mathbf{I}_{\epsilon1}^{-1}) )
\end{equation}
and:
\begin{equation}
\mathbf{l}_2 = \mathbf{d}_2 \mathbf{I}_{\epsilon}^{-1} - (\mathbf{x}_2\cdot (\mathbf{d}_2 \mathbf{I}_{\epsilon1}^{-1}) )
\end{equation}
Both $\mathbf{d}_1 \mathbf{I}_{\epsilon1}$ and $\mathbf{d}_2 \mathbf{I}_{\epsilon1}$ are $2$-vectors, therefore, $\mathbf{L}$ is a $4$-vector.  Note that a line can similarly be obtained by the intersection of two DCGA planes as:
\begin{equation}
\mathbf{L} = \mathbold{\Pi}_{1} \wedge \mathbold{\Pi}_2
\end{equation}
Finally, the intersection is computed as:
\begin{equation}
\mathbf{P}_p = \mathbf{Q}_{DCGA} \wedge \mathbf{L}
\end{equation}
~\\ \textbf{Complexity of some major operations of DCGA}~\\
Let us first evaluate the computational cost of checking whether a point is on a quadric using the model of~\ref{sec:quadricComparisonModel}. $\mathbf{Q}_{DCGA}$ has a total of $10$ basis bivector components. For each basis bivector, at most $3$ inner products (bivector $\wedge$ bivector) are performed. Finally, the number of point component is $25$. Thus, the product $\mathbf{Q}_{DCGA} \cdot \mathbf{X}$  require $25 \times 3 \times 10=750$ products. 

~\\The cost of the computation of the tangent plane to a quadric corresponds to the cost of, first, the normal vector $\mathbf{n}_1$ and second the tangent plane. Firstly, the normal vector is defined as:
\begin{equation}
\mathbf{n}_{1} = ((\mathbf{D}_{x} \times \mathbf{Q}_{DCGA}) \cdot \mathbf{x}) \mathbf{e}_1 + ((\mathbf{D}_{y} \times \mathbf{Q}_{DCGA}) \cdot \mathbf{x}) \mathbf{e}_2 + ((\mathbf{D}_{z} \times \mathbf{Q}_{DCGA}) \cdot \mathbf{x}) \mathbf{e}_3 
\end{equation}
Equation~\eqref{eq:differentialOperators} defined $\mathbf{D}_{x},\mathbf{D}_{y},\mathbf{D}_{z}$, and the commutator product of these operators with the quadric results in a $7$-component bivector. Indeed, the extractions operators $\mathbf{T}_x,\mathbf{T}_y,\mathbf{T}_z$ are 2-components operator while $\mathbf{T}_1$ is a single component extraction operator. Each inner product with $\mathbf{X}$ then has a computational cost of $7 \times 25=175$ products. This latter computation is repeated for each axis thus this results in $175 \times 3 = 525$ products.

~\\Second, the tangent plane is obtained by using the normal vector and the orthogonal distance from the origin. The orthogonal distance is computed as:
\begin{equation}
   \mathrm{d} = \mathbf{n}_1 \cdot \mathbf{x}_1
\end{equation} 
Thus, this requires the computation of $3$ inner products. Finally, the tangent plane is the results of the outer product:
\begin{equation}
\mathbold{\Pi} = (\mathbf{n}_{1} + \mathrm{d} \mathbf{e}_{\infty1} ) \wedge (\mathbf{n}_{2} + \mathrm{d} \mathbf{e}_{\infty2} )
\end{equation}
Both $(\mathbf{n}_{1} + \mathrm{d} \mathbf{e}_{\infty1} ) $ and $(\mathbf{n}_{2} + \mathrm{d} \mathbf{e}_{\infty2} ) $ are $4$-components 1-vector. Thus, the computational cost of the outer product is $4 \times 4=16$. Hence, the total cost of the computation of the tangent plane is $525+16=541$ products.

~\\ The cost of the computation of the intersection between a quadric and a line consists in evaluating the cost of the outer product between a DCGA line $\mathbf{l}$ and the quadric of DCGA $\mathbf{q}_{DCGA}$. In the previous section, we defined a line $\mathbf{l}$ as the 4-vector entity obtained by the outer product of the planes as follows:
\begin{equation}
\mathbf{L} = \mathbold{\Pi}_{1} \wedge \mathbold{\Pi}_2
\end{equation}
A plane in DCGA is obtained as the outer product of two Conformal Geometric Algebra (CGA) planes whose number of components is $4$. The result of the outer product between the two CGA planes may have non-zero components along the following components:
\begin{equation}
\begin{array}{c@{}l}
&(\mathbf{e}_{14},\mathbf{e}_{15},\mathbf{e}_{16},\mathbf{e}_{1\infty 2},
 \mathbf{e}_{24},\mathbf{e}_{25},\mathbf{e}_{26},\mathbf{e}_{2\infty 2},
 \mathbf{e}_{34},\\
& \mathbf{e}_{35},\mathbf{e}_{36},\mathbf{e}_{3\infty 2},
  \mathbf{e}_{\infty 1 4},\mathbf{e}_{\infty 1 5},\mathbf{e}_{\infty 1 6},
  \mathbf{e}_{\infty1 4},\mathbf{e}_{\infty1 \infty2})
\end{array}
\end{equation}
Then, computing the outer product between two planes may have some results along the following basis quad-vectors:
\begin{equation}
\begin{array}{c@{}l}
(&\mathbf{e}_{1245},\mathbf{e}_{1246},\mathbf{e}_{124\infty 2},\mathbf{e}_{1256},\mathbf{e}_{125\infty 2},\mathbf{e}_{126\infty 2},\mathbf{e}_{1345},\mathbf{e}_{1346},\mathbf{e}_{134\infty 2},\mathbf{e}_{1356},\mathbf{e}_{135\infty2},\\
&\mathbf{e}_{136\infty2},\mathbf{e}_{1\infty145},\mathbf{e}_{1\infty146},\mathbf{e}_{1\infty 14\infty 2},\mathbf{e}_{1\infty 156},\mathbf{e}_{1\infty 1 5 \infty 2},\mathbf{e}_{1\infty 1 6 \infty 2},\mathbf{e}_{2345},\mathbf{e}_{2346},\mathbf{e}_{234\infty 2},\\
&\mathbf{e}_{2356},\mathbf{e}_{235\infty2},\mathbf{e}_{236\infty2},\mathbf{e}_{2\infty145},\mathbf{e}_{2\infty146},\mathbf{e}_{2\infty14\infty2},\mathbf{e}_{2\infty156},\mathbf{e}_{2\infty 1 5\infty2},\mathbf{e}_{2\infty 1 6 \infty 2},\mathbf{e}_{3\infty 1 4 5},\\
&\mathbf{e}_{3\infty 1 4  6},\mathbf{e}_{3\infty 1 4 \infty 2},\mathbf{e}_{3\infty 1 56},\mathbf{e}_{3\infty 1 5 \infty 2},\mathbf{e}_{3\infty 1 6 \infty 2}) 
\end{array}
\end{equation}
This latter $4$-vector has thus $36$ components. Then, the outer product between this quad-vector $\mathbf{L}$ and the quadratic surface can be performed as:
\begin{equation}
\mathbf{Q}_{DCGA} \wedge \mathbf{L}
\end{equation}
As the number of components of $\mathbf{Q}_{DCGA}$ is $25$ and the number of components of $\mathbf{L}$ is $36$. Then the cost of the outer product is $25 \times 36=900$ products. The following table summarises the computational cost of the three features computed so far for DCGA.
\begin{table}[ht]
\caption{Computational features in number of Geometric Algebra operations for DCGA}
\begin{center}
\renewcommand{\arraystretch}{1.3}
\begin{tabular}{|c|c|}
\hline
Feature & \textbf{DCGA} \\
\hline
point is on a quadric & $725$ \\
\hline
tangent plane & $541$ \\
\hline
quadric-line intersection & $900$ \\
\hline
\end{tabular}
\end{center}
\end{table}

\subsubsection{DPGA of $\mathbb{G}_{4,4}$}
DPGA was adapted from the approach of Parkin~\cite{Parkin} in $2012$ and firstly introduced in 2015 by Goldman and Mann~\cite{Goldman2015} and further developed by Du and Goldman and Mann~\cite{Goldman2017}.
~\\ \textbf{Basis and metric}~\\
DPGA $\mathbb{G}_{4,4}$ is defined over a 8-dimensional vector space. In a similar way to DCGA, the base vectors of the space are divided into two groups: $ \{ \mathbf{w}_0,\mathbf{w}_1,\mathbf{w}_2,\mathbf{w}_3\} $ (corresponding to the projective Geometric Algebra vectors), and a copy of this basis $ \{ \mathbf{w}^{*}_0,\mathbf{w}^{*}_1,\mathbf{w}^{*}_2,\mathbf{w}^{*}_3\} $ such that $\mathbf{w}_i \mathbf{w}^{*}_i=0.5 + \mathbf{w}_i \wedge \mathbf{w}^{*}_i$, $\forall i \in \{0,1,2,3\}$. To have more details, we show the inner products between any basis vectors in Table~\ref{table:metricDPGA}.
\begin{table}[!ht]
\caption{Inner product between DPGA basis vectors.}
$$
\begin{array}{c|cccccccc}
      & \mathbf{w_0} & \mathbf{w_1} & \mathbf{w_2} & \mathbf{w_3} & \mathbf{w^{*}_0} & \mathbf{w^{*}_1} & \mathbf{w^{*}_2} & \mathbf{w^{*}_3} \\ \hline 
 
\mathbf{w_0} & \textcolor{gray}{0} & \textcolor{gray}{0} & \textcolor{gray}{0} & \textcolor{gray}{0} & \cellcolor{verylightgray} 0.5 & \cellcolor{verylightgray}\textcolor{gray}{0} &  \cellcolor{verylightgray}\textcolor{gray}{0} & \cellcolor{verylightgray}\textcolor{gray}{0} \\

\mathbf{w_1} & \textcolor{gray}{0} & \textcolor{gray}0 & \textcolor{gray}{0} & \textcolor{gray}{0} & \cellcolor{verylightgray}\textcolor{gray}{0} & \cellcolor{verylightgray}0.5 & \cellcolor{verylightgray}\textcolor{gray}{0} & \cellcolor{verylightgray}\textcolor{gray}{0} \\

\mathbf{w_2} & \textcolor{gray}{0} & \textcolor{gray}{0} & \textcolor{gray}0 & \textcolor{gray}{0} & \cellcolor{verylightgray}\textcolor{gray}{0} & \cellcolor{verylightgray}\textcolor{gray}{0} & \cellcolor{verylightgray}0.5 & \cellcolor{verylightgray}\textcolor{gray}{0} \\

\mathbf{w_3} & \textcolor{gray}{0} & \textcolor{gray}{0} & \textcolor{gray}{0} & \textcolor{gray}0 & \cellcolor{verylightgray}\textcolor{gray}{0} & \cellcolor{verylightgray}\textcolor{gray}{0} & \cellcolor{verylightgray}\textcolor{gray}{0} & \cellcolor{verylightgray}0.5\\

\mathbf{w^{*}_0}  & \cellcolor{verylightgray} 0.5 & \cellcolor{verylightgray}\textcolor{gray}{0} & \cellcolor{verylightgray}\textcolor{gray}{0} & \cellcolor{verylightgray}\textcolor{gray}{0} & \textcolor{gray}{0} & \textcolor{gray}{0} & \textcolor{gray}{0} & \textcolor{gray}{0} \\
 
\mathbf{w^{*}_1} & \cellcolor{verylightgray}\textcolor{gray}{0} & \cellcolor{verylightgray}0.5 & \cellcolor{verylightgray}\textcolor{gray}{0} & \cellcolor{verylightgray}\textcolor{gray}{0} & \textcolor{gray}{0} & \textcolor{gray}0 & \textcolor{gray}0 & \textcolor{gray}0\\
 
\mathbf{w^{*}_2} & \cellcolor{verylightgray}\textcolor{gray}{0} & \cellcolor{verylightgray}\textcolor{gray}{0} & \cellcolor{verylightgray}0.5 & \cellcolor{verylightgray}\textcolor{gray}{0} & \textcolor{gray}{0} & \textcolor{gray}0 & \textcolor{gray}0 & \textcolor{gray}0 \\
 
\mathbf{w^{*}_3} & \cellcolor{verylightgray}\textcolor{gray}{0} & \cellcolor{verylightgray}\textcolor{gray}{0} & \cellcolor{verylightgray}\textcolor{gray}{0} & \cellcolor{verylightgray}0.5 & \textcolor{gray}{0} & \textcolor{gray}0 &\textcolor{gray}0 &\textcolor{gray}0 \\

\end{array} 
$$
\label{table:metricDPGA}
\end{table}

~\\ \textbf{Point of DPGA}~\\
In DPGA, the entity representing a point whose Euclidean coordinates are $(x,y,z)$ has two definitions, namely a primal and dual. Both of the definitions are the base to construct quadrics by means of sandwiching product. The definitions of the points are:
\begin{equation}
\begin{array}{cc}
\mathbf{p} &= x \mathbf{w}_0 + y \mathbf{w}_1 + z\mathbf{w}_2 + w\mathbf{w}_3 \\
\mathbf{p}^* &= x\mathbf{w}^*_0 + y\mathbf{w}^*_1 + z\mathbf{w}^*_2 + w\mathbf{w}^*_3 \\
\end{array}
\end{equation}
Note that the dual definition denotes the fact that:
\begin{equation}
\mathbf{w}_i \cdot \mathbf{w}^*_j = \displaystyle \frac{1}{2} \delta_{i,j} \text{ } \forall i,j=0,\cdots 3
\end{equation}
Where $\delta_{i,j}=1$ if $i=j$, else $0$. This corresponds to the condition of the dual stated in Section 11 of~\cite{doran1993lie}. 
~\\ \textbf{Quadrics}~\\
Again, for a general quadric defined as:
\begin{equation}
\mathrm{a} x^2+\mathrm{b} y^2 +\mathrm{c}z^2+\mathrm{d}xy+\mathrm{e}yz+\mathrm{f}zx+\mathrm{g}x+\mathrm{h} y  + \mathrm{i}z + \mathrm{j} =0
\end{equation}
A quadric in DPGA is the bivector $Q_{DPGA}$ defined as follows:
\begin{equation}
\begin{array}{c@{}l}
Q_{DPGA} =& 4\mathrm{a} \mathbf{w}_0^{*} \wedge \mathbf{w}_0 + 4\mathrm{b} \mathbf{w}_1^{*} \wedge \mathbf{w}_1 + 4\mathrm{c} \mathbf{w}_2^{*} \wedge \mathbf{w}_2 + 4\mathrm{j} \mathbf{w}_3^{*} \wedge \mathbf{w}_3 \\
&+ 2\mathrm{d}( \mathbf{w}_0^{*} \wedge \mathbf{w}_1+ \mathbf{w}_1^{*} \wedge \mathbf{w}_0) 
+ 2\mathrm{e} (\mathbf{w}_0^{*} \wedge \mathbf{w}_2 + \mathbf{w}_2^{*} \wedge \mathbf{w}_0) \\
&+ 2\mathrm{f} ( \mathbf{w}_1^{*} \wedge \mathbf{w}_2 + \mathbf{w}_2^{*} \wedge \mathbf{w}_1) + 2\mathrm{g} (\mathbf{w}_0^{*} \wedge \mathbf{w}_3+ \mathbf{w}_3^{*} \wedge \mathbf{w}_0) \\
&+ 2\mathrm{h} (\mathbf{w}_1^{*} \wedge \mathbf{w}_3 + \mathbf{w}_3^{*} \wedge \mathbf{w}_1) + 2\mathrm{i}(\mathbf{w}_2^{*} \wedge \mathbf{w}_3 + \mathbf{w}_3^{*} \wedge \mathbf{w}_2) 
\end{array}
\end{equation}

Finally, a point $(x,y,z)$ is in the quadric $\mathbf{Q}_{DPGA}$ if and only if 
\begin{equation}
\mathbf{p} \cdot \mathbf{Q}_{DPGA} \cdot \mathbf{p}^* =0
\end{equation}
Let us call $f_{DPGA} = \mathbf{p} \cdot \mathbf{Q}_{DPGA} \cdot \mathbf{p}^*$. Then to investigate numerical properties of the quadric computation, we develop the formula $ \mathbf{p} \cdot \mathbf{Q}_{DPGA} \cdot \mathbf{p}^*$:
\begin{equation}
\begin{array}{ccl}
f_{DPGA} &= &\mathbf{p} \cdot \mathbf{Q}_{DPGA} \cdot \mathbf{p}^* \\
&=& \Big( 2\mathrm{a} x \mathbf{w}_0 + \mathrm{d} x \mathbf{w}_1 + \mathrm{e} x \mathbf{w}_2 + \mathrm{g} x \mathbf{w}_3 + 2\mathrm{b} y \mathbf{w}_1 + \mathrm{d} y \mathbf{w}_0 + \mathrm{f} y \mathbf{w}_2 + \mathrm{h} y \mathbf{w}_3 \\
&&+ 2\mathrm{c} z \mathbf{w}_2 + \mathrm{e} z \mathbf{w}_0 + \mathrm{f} z \mathbf{w}_1 + \mathrm{i} z \mathbf{w}_3 
+ 2\mathrm{j} \mathbf{w}_3 + \mathrm{g}  \mathbf{w}_0 + \mathrm{h}  \mathbf{w}_1 + \mathrm{i} \mathbf{w}_2\Big) \cdot \mathbf{p}^*  \\
&=&\Big( (2\mathrm{a} x + \mathrm{d} y + \mathrm{e} z + \mathrm{g})\mathbf{w}_0 + (2\mathrm{b} y + \mathrm{d} x + \mathrm{f} z + \mathrm{h})\mathbf{w}_1 \\
&&+(2\mathrm{c} z + \mathrm{e} x + \mathrm{f} y + \mathrm{i})\mathbf{w}_2 + (\mathrm{i}z + \mathrm{g} x + \mathrm{h} y + 2\mathrm{j})\mathbf{w}_3) \Big) \cdot \mathbf{p}^*
\end{array}
\end{equation}
The last inner product results in:
\begin{equation}
\begin{array}{ccl}
f_{DPGA} &= &\mathbf{p} \cdot \mathbf{Q}_{DPGA}  \cdot \mathbf{p}^* \\
&=&\mathrm{a} x^2 + 0.5\mathrm{d} xy + 0.5\mathrm{e} xz + 0.5\mathrm{g}x + \mathrm{b} y^2 + 0.5\mathrm{d} xy + 0.5\mathrm{f} yz + 0.5\mathrm{h}y \\
&&+\mathrm{c} z^2 + 0.5\mathrm{e} xz + 0.5\mathrm{f} yz + 0.5\mathrm{i}z + 0.5\mathrm{i}z + 0.5\mathrm{g} x + 0.5\mathrm{h} y + \mathrm{j}\\
&=&\mathrm{a} x^2 + \mathrm{b} y^2 + \mathrm{c} z^2 + \mathrm{d} xy + \mathrm{e} xz + \mathrm{f} yz + \mathrm{g}x + \mathrm{h}y + \mathrm{i}z + \mathrm{j}\\
\end{array}
\end{equation}

This latter development is the base to determine the number of operations in the computation of $\mathbf{p} \cdot \mathbf{Q}_{DPGA} \cdot \mathbf{p}^* $. 

~\\\textbf{Plane tangent to a quadric}~\\
In a similar way as CGA, DPGA supports the computation of the tangent plane $\mathbold{\Pi}^*$ to a quadric $\mathbf{Q}_{DPGA}$ at a dual point $\mathbf{p}^*$ as follows:
\begin{equation}
\mathbold{\Pi}^* = \mathbf{Q}_{DPGA} \cdot \mathbf{p}^*
\end{equation} 
~\\\textbf{Quadric-line intersection}~\\
DPGA also supports the quadric-line intersection. Given a line defined as the primal and dual $\mathbf{L}= \mathbf{x}_1 \wedge \mathbf{x}_2$ and $\mathbf{L}^* = \mathbf{x}_1^* \wedge \mathbf{x}_2^*$  and the quadratic surface bivector $\mathbf{Q}_{DPGA}$. The intersection bivector called $\mathbf{P}_p$ is defined as follows:
\begin{equation}
\mathbf{P}_p = (\mathbf{L}^* \wedge \mathbf{Q}_{DPGA} \wedge \mathbf{L}) \cdot \mathbf{I}
\end{equation}
where $\mathbf{I}$ is the pseudo-scalar of $\mathbb{G}_{4,4}$ defined as:
\begin{equation}
\mathbf{I} = \mathbf{w}_{0} \wedge \mathbf{w}_{1} \wedge \mathbf{w}_{2} \wedge \mathbf{w}_{3} \wedge \mathbf{w}_{0}^* \wedge \mathbf{w}_{1}^*\wedge \mathbf{w}_{2}^* \wedge \mathbf{w}_{3}^*
\end{equation}
~\\ \textbf{Complexity of some major operations of DPGA}~\\
$\mathbf{Q}_{DPGA}$ has a total of $16$ basis bivector components. For each basis bivector, $2$ inner products are performed. Thus, the first product $\mathbf{P} \cdot \mathbf{Q}_{DCGA}$ will  require $4 \times 2 \times 16=128$ inner products. As previously seen, the resulting entity is a vector with $4$ components. Hence, the second inner product requires $4 \times 4=16$ products. This results in $144$ products. 

~\\ Let us now evaluate the cost of the intersection between a quadratic surface $\mathbf{Q}_{DPGA}$ and a line $\mathbf{L}$ and $\mathbf{L}^*$. The line $\mathbf{L}^*$ is obtained by the outer product of two points $\mathbf{x}_1$ and $\mathbf{x}_2$ whose number of components is $4$. Thus a line $\mathbf{L}$ has $6$ components along the bivector basis:
\begin{equation}
(\mathbf{w}_{0} \wedge \mathbf{w}_{1},\mathbf{w}_{0} \wedge \mathbf{w}_{2},\mathbf{w}_{0} \wedge \mathbf{w}_{3},\mathbf{w}_{1} \wedge \mathbf{w}_{2},\mathbf{w}_{1} \wedge \mathbf{w}_{3},\mathbf{w}_{2} \wedge \mathbf{w}_{3})
\end{equation}
and a line $\mathbf{L}^*$ has the following bivector basis components:
\begin{equation}
(\mathbf{w}_{0}^* \wedge \mathbf{w}_{1}^*,\mathbf{w}_{0}^* \wedge \mathbf{w}_{2}^*,\mathbf{w}_{0}^* \wedge \mathbf{w}_{3}^*,\mathbf{w}_{1}^* \wedge \mathbf{w}_{2}^*,\mathbf{w}_{1}^* \wedge \mathbf{w}_{3}^*,\mathbf{w}_{2}^* \wedge \mathbf{w}_{3}^*)
\end{equation}
The number of components of the quadratic surface is $16$ and the number of components of the line is $6$. Then, the computational cost of the outer product $\mathbf{L}^* \wedge \mathbf{Q}_{DPGA}$ is $6 \times 16=96$ outer products. The result is a $4$-vector and may have components along the quad-vector basis:
\begin{equation}
\begin{array}{c}
(\mathbf{w}_0^{*} \wedge \mathbf{w}_1^{*} \wedge \mathbf{w}_2^{*} \wedge \mathbf{w}_0,\mathbf{w}_0^{*} \wedge \mathbf{w}_1^{*} \wedge \mathbf{w}_2^{*} \wedge \mathbf{w}_1,\mathbf{w}_0^{*} \wedge \mathbf{w}_1^{*} \wedge \mathbf{w}_2^{*} \wedge \mathbf{w}_2,\\
\mathbf{w}_0^{*} \wedge \mathbf{w}_1^{*} \wedge \mathbf{w}_2^{*} \wedge \mathbf{w}_3,  
\mathbf{w}_0^{*} \wedge \mathbf{w}_1^{*} \wedge \mathbf{w}_3^{*} \wedge \mathbf{w}_0,
\mathbf{w}_0^{*} \wedge \mathbf{w}_1^{*} \wedge \mathbf{w}_3^{*} \wedge \mathbf{w}_1,\\
\mathbf{w}_0^{*} \wedge \mathbf{w}_1^{*} \wedge \mathbf{w}_3^{*} \wedge \mathbf{w}_2,\mathbf{w}_0^{*} \wedge \mathbf{w}_1^{*} \wedge \mathbf{w}_3^{*} \wedge \mathbf{w}_3,
\mathbf{w}_0^{*} \wedge \mathbf{w}_2^{*} \wedge \mathbf{w}_3^{*} \wedge \mathbf{w}_0,\\
\mathbf{w}_0^{*} \wedge \mathbf{w}_2^{*} \wedge \mathbf{w}_3^{*} \wedge \mathbf{w}_1,
\mathbf{w}_0^{*} \wedge \mathbf{w}_2^{*} \wedge \mathbf{w}_3^{*} \wedge \mathbf{w}_2,
\mathbf{w}_0^{*} \wedge \mathbf{w}_2^{*} \wedge \mathbf{w}_3^{*} \wedge \mathbf{w}_3, \\
\mathbf{w}_1^{*} \wedge \mathbf{w}_2^{*} \wedge \mathbf{w}_3^{*} \wedge \mathbf{w}_0,
\mathbf{w}_1^{*} \wedge \mathbf{w}_2^{*} \wedge \mathbf{w}_3^{*} \wedge \mathbf{w}_1,
\mathbf{w}_1^{*} \wedge \mathbf{w}_2^{*} \wedge \mathbf{w}_3^{*} \wedge \mathbf{w}_2,\\
\mathbf{w}_1^{*} \wedge \mathbf{w}_2^{*} \wedge \mathbf{w}_3^{*} \wedge \mathbf{w}_3) 
\end{array}
\end{equation}
Thus, the resulting entity has $16$ components. Furthermore, the line $\mathbf{l}$ has 6 components. Hence, the cost of the final outer product is $16 \times 6=96$ outer products. Finally, the total operation cost is thus $96+96=192$ products.  

~\\ The cost of the computation of the tangent plane at a point $\mathbf{p}$ is the cost of the following product:
\begin{equation}
\mathbold{\pi}^* = \mathbf{Q}_{DPGA} \cdot \mathbf{p}^*
\end{equation} 
Considering the fact that the number of components of $\mathbf{p}^*$ is $4$ and the number of components of $\mathbf{Q}_{DPGA}$ is $16$. Then the computational cost of the computation of the tangent plane is $16 \times 4=64$ products.

\subsubsection{QCGA of $\mathbb{G}_{9,6}$}
Let us evaluate the computational cost of checking whether a point is on a quadratic surface. $\mathbf{Q}^{*}$ has a total of $12$ basis vector components. For each basis vector, at most $1$ inner product is performed, please refer to the Equation~\eqref{eq:innerProductGeneral}. Finally, the number of point component is $12$. Thus, the product $\mathbf{x} \cdot \mathbf{Q}^*$ requires at most $12 \times 12=144$ products.

~\\ The computation of the tangent plane is performed by first the computation of the normal vector as:
\begin{align}
    \mathbf{n}_{\epsilon} =& \Big(\big((\mathbf{x} \cdot \mathbf{e}_{1}) \mathbf{e}_{\infty1} + (\mathbf{x} \cdot \mathbf{e}_{2}) \mathbf{e}_{\infty4} + (\mathbf{x} \cdot \mathbf{e}_{3}) \mathbf{e}_{\infty5} +\mathbf{e}_{1}\big) \cdot \mathbf{q}^*\Big) \mathbf{e}_1 + \nonumber \\
    &\Big(\big((\mathbf{x} \cdot \mathbf{e}_{2}) \mathbf{e}_{\infty2} + (\mathbf{x} \cdot \mathbf{e}_{1}) \mathbf{e}_{\infty4} + (\mathbf{x} \cdot \mathbf{e}_{3}) \mathbf{e}_{\infty6} +\mathbf{e}_{2}\big) \cdot \mathbf{q}^*\Big) \mathbf{e}_2 + \nonumber\\
    &\Big(\big((\mathbf{x} \cdot \mathbf{e}_{3}) \mathbf{e}_{\infty3} + (\mathbf{x} \cdot \mathbf{e}_{1}) \mathbf{e}_{\infty5} + (\mathbf{x} \cdot \mathbf{e}_{2}) \mathbf{e}_{\infty6} +\mathbf{e}_{3}\big) \cdot \mathbf{q}^*\Big) \mathbf{e}_3.
\end{align} 
This computation required the inner product between a vector with $12$ components and another vector with $4$ components. This computation is repeated for each Euclidean basis vector thus the computation of the normal vector requires $3 \times 4 \times 12=144$ inner products. 

~\\Then, the tangent plane is computed using the normal vector as follows:
\begin{equation}
    \mathbold{\pi^{*}} = \mathbf{n}_{\epsilon} + \frac{1}{3} \big(\mathbf{e}_{\infty1} + \mathbf{e}_{\infty2} + \mathbf{e}_{\infty3}\big) \sqrt{-2(\mathbf{e}_{o1} + \mathbf{e}_{o2} + \mathbf{e}_{o3}) \cdot \mathbf{x}}.
\end{equation}  
This computation requires the computation of an inner product of a vector with $3$ components ($\mathbf{e}_{1},\mathbf{e}_{2},\mathbf{e}_{3}$) with a $12$ component-vector. This means $12 \times 3=36$ products. Thus the total number of inner products required in the computation of the tangent plane is $144+36=180$ products.

~\\ The final computational feature is the quadric-line intersection. In QCGA, this simply consists in computing the outer product:
\begin{equation}
\mathbf{c}^* = \mathbf{q}^* \wedge \mathbf{l}^*
\end{equation}
The number of components of $\mathbf{q}^*$ is $12$ as already seen. In QCGA, we defined a line with the $6$ Plücker coefficients as:
\begin{equation}
\mathbf{l}^* =  3 \,\mathbf{m} \, \mathbf{I}_{\epsilon} + (\mathbf{e}_{\infty 3} +\mathbf{e}_{\infty 2} +\mathbf{e}_{\infty 1}) \wedge \, \mathbf{n}\, \mathbf{I}_{\epsilon} .
\end{equation}
The number of components of both $\mathbf{m}$ and $\mathbf{n}$ is $3$. The outer product $(\mathbf{e}_{\infty 3} +\mathbf{e}_{\infty 2} +\mathbf{e}_{\infty 1}) \wedge \, \mathbf{n}\, \mathbf{I}_{\epsilon}$ yields a copy of the $3$ components of $\mathbf{n}$ along $\mathbf{e}_{\infty 1},\mathbf{e}_{\infty 2},\mathbf{e}_{\infty 3}$ basis vectors. Thus, the number of components of $\mathbf{l}^*$ is $3 \times 3 + 3 = 12$. Finally, the cost of the outer product between $\mathbf{q}^*$ and $\mathbf{l}^*$ is $12 \times 12=144$ products.  

~\\ The table~\ref{tab:comparisonModels} summarizes the computational features of the proposed framework compared to DPGA and DCGA.
\begin{table}[ht]
\caption{Comparison of the computational features between QCGA, DCGA, DPGA in number of Geometric Algebra operations}
\label{tab:comparisonModels}
\begin{center}
\renewcommand{\arraystretch}{1.3}
\begin{tabular}{|c|c|c|c|}
\hline
Feature & DPGA & DCGA & \textbf{QCGA} \\
\hline
point is on a quadric & $\mathbf{144}$ & $750$ & $\mathbf{144}$ \\
\hline
tangent plane & $\mathbf{64}$ & $541$ & $180$ \\
\hline
quadric-line intersection & $192$ & $300$ & $\mathbf{144}$ \\
\hline
\end{tabular}
\end{center}
\end{table}
We remark that the computation of the tangent plane is more efficient using DPGA whereas the intersection between a quadratic surface and a line requires less computations using QCGA. Furthermore, some versors are not defined in some models.

\section{Mapping}
As a practical application, it might be interesting to construct a quadratic surface from 9 points then rotating this quadratic surface. For the moment, QCGA is the only approach, in Geometric Algebra, that can construct quadratic surface from 9 points but it does not yet support all the transformations. Furthermore, the last chapter shows that some operations are worth doing in a certain framework. These points are our motivation for defining new operators that convert quadratic surfaces between the three presented frameworks. 

~\\The key idea is that for any entities representing quadric surface in QCGA,DCGA and DPGA, it is possible to convert the entity such that all the coefficients of the quadrics:
\begin{equation}
\mathrm{a} x^2+\mathrm{b} y^2 +\mathrm{c}z^2+\mathrm{d}xy+\mathrm{e}yz+\mathrm{f}zx+\mathrm{g}x+\mathrm{h} y  + \mathrm{i}z + \mathrm{j} =0,
\label{eq:genralQuadrics}
\end{equation} 
can be extracted easily.

\section{DCGA reciprocal operators}
This means defining reciprocal operators for DCGA:
\begin{equation}
\begin{array}{c@{}ll@{}l}
\mathbf{T}^{x^{2}} &= \mathbf{e}_1 \wedge \mathbf{e}_4 & \mathbf{T}_{y^{2}} &= \mathbf{e}_2 \wedge \mathbf{e}_5\\
\mathbf{T}^{z^{2}} &= \mathbf{e}_3 \wedge \mathbf{e}_6 & \mathbf{T}_{1} &= \mathbf{e}_{o1} \wedge \mathbf{e}_{o2} \\
\end{array}
\end{equation}
along with the $6$ following:
\begin{equation}
\begin{array}{cl}
\mathbf{T}^{x} &= \Big( \mathbf{e}_1 \wedge \mathbf{e}_{o 2} + \mathbf{e}_{o 1} \wedge \mathbf{e}_4 \Big)\\
 \mathbf{T}^{y} &=  \Big( \mathbf{e}_2 \wedge \mathbf{e}_{o 2} + \mathbf{e}_{o 1} \wedge \mathbf{e}_5 \Big) \\
\mathbf{T}^{z} &= \Big( \mathbf{e}_3 \wedge \mathbf{e}_{o 2} + \mathbf{e}_{o 1} \wedge \mathbf{e}_6 \Big) \\ 
\mathbf{T}^{xy} &= \Big(\mathbf{e}_1 \wedge \mathbf{e}_5 + \mathbf{e}_2 \wedge \mathbf{e}_4 \Big)\\
\mathbf{T}^{xz} &= \Big(\mathbf{e}_1 \wedge \mathbf{e}_6 + \mathbf{e}_3 \wedge \mathbf{e}_4 \Big)\\
\mathbf{T}^{yz} &= \Big(\mathbf{e}_3 \wedge \mathbf{e}_5 + \mathbf{e}_2 \wedge \mathbf{e}_6 \Big)\\
\end{array}
\end{equation}

These reciprocal operators verify the following properties:
\begin{equation}
\begin{array}{l@{}l@{}l@{}l@{}l@{}l@{}}
\mathbf{T}^{x^2} \cdot \mathbf{T}_{x^2} &= 1,~
\mathbf{T}^{y^2} \cdot \mathbf{T}_{y^2} &= 1,~
\mathbf{T}^{z^2} \cdot \mathbf{T}_{z^2} &= 1,~
\mathbf{T}^{xy} \cdot \mathbf{T}_{xy}   &= 1,~
\mathbf{T}^{xz} \cdot \mathbf{T}_{xz}   &= 1,\\
\mathbf{T}^{yz} \cdot \mathbf{T}_{yz}   &= 1,~
\mathbf{T}^{x} \cdot \mathbf{T}_{x}     &= 1,~
\mathbf{T}^{y} \cdot \mathbf{T}_{y}     &= 1,~
\mathbf{T}^{z} \cdot \mathbf{T}_{z}     &= 1,~
\mathbf{T}^{1} \cdot \mathbf{T}_{1}     &= 1\\
\end{array}
\end{equation}
Then, given $\mathbf{q}_{DCGA}$ the entity representing a quadratic surface of DCGA, any coefficients of this quadratic surface~\eqref{eq:genralQuadrics} can be extracted as:
\begin{equation}
\begin{array}{l@{}l@{}l@{}l@{}l@{}l@{}}
\mathbf{T}^{x^2} \cdot \mathbf{q}_{DCGA} &= \mathrm{a},~
\mathbf{T}^{y^2} \cdot \mathbf{q}_{DCGA} &= \mathrm{b},~
\mathbf{T}^{z^2} \cdot \mathbf{q}_{DCGA} &= \mathrm{c},~
\mathbf{T}^{xy} \cdot \mathbf{q}_{DCGA}   &= \mathrm{d},\\
\mathbf{T}^{xz} \cdot \mathbf{q}_{DCGA}   &= \mathrm{e},
\mathbf{T}^{yz} \cdot \mathbf{q}_{DCGA}  &= \mathrm{f},~
\mathbf{T}^{x} \cdot \mathbf{q}_{DCGA}    &= \mathrm{g},~
\mathbf{T}^{y} \cdot \mathbf{q}_{DCGA}     &= \mathrm{h},\\
\mathbf{T}^{z} \cdot \mathbf{q}_{DCGA}     &= \mathrm{i},~
\mathbf{T}^{1} \cdot \mathbf{q}_{DCGA}    &= \mathrm{j}\\
\end{array}
\end{equation}
The coefficients of the implicit form of the surface can then be extracted as:
\begin{equation}
\begin{array}{@{}lll}

\mathrm{a} = \mathbf{T}^{x^2} \cdot \mathbf{q}_{DCGA} & 
\mathrm{b} = \mathbf{T}^{y^2} \cdot \mathbf{q}_{DCGA} &
\mathrm{c} = \mathbf{T}^{z^2} \cdot \mathbf{q}_{DCGA}\\

\mathrm{d} = \mathbf{T}^{xy} \cdot \mathbf{q}_{DCGA} &
\mathrm{e} = \mathbf{T}^{xz} \cdot \mathbf{q}_{DCGA} &
\mathrm{f} = \mathbf{T}^{yz} \cdot \mathbf{q}_{DCGA}\\

\mathrm{g} = \mathbf{T}^{x} \cdot \mathbf{q}_{DCGA} &
\mathrm{h} = \mathbf{T}^{y} \cdot \mathbf{q}_{DCGA} &
\mathrm{i} = \mathbf{T}^{z} \cdot \mathbf{q}_{DCGA} \\
~\mathrm{j} = \mathbf{T}^{1} \cdot \mathbf{q}_{DCGA} \\
\end{array}
\end{equation}
The construction of a DCGA point was already explained in the previous section and defined in~\cite{easter20162CGA}. 
The reciprocal operation requires the computation of the normalization point $\mathbf{\hat{x}}$ of DCGA that we define as: 
\begin{equation}
\mathbf{\hat{x}} = \frac{\mathbf{x}}{\mathbf{x} \cdot (\mathbf{e}_{\infty1} \wedge  \mathbf{e}_{\infty2})} 
\end{equation}
The extraction of the Euclidean components $(x,y,z)$ of a normalized point $\mathbf{\hat{x}}$ of DCGA can be performed as follows:
\begin{equation}
x = \mathbf{\hat{x}} \cdot \mathbf{e}_1, ~~~~~ 
y = \mathbf{\hat{x}} \cdot \mathbf{e}_2, ~~~~~
z = \mathbf{\hat{x}} \cdot \mathbf{e}_3
\end{equation}

\section{DPGA reciprocal operators}
In a similar way, let us note $\mathbf{W}$ reciprocal operators for DPGA
\begin{equation}
\begin{array}{c@{}lc@{}l}
\mathbf{W}^{x^2} &=  \mathbf{w}_{0}^{*} \wedge \mathbf{w}_{0} & 
\mathbf{W}^{y^2} &=  \mathbf{w}_{1}^{*} \wedge \mathbf{w}_{1} \\
\mathbf{W}^{z^2} &=  \mathbf{w}_{2}^{*} \wedge \mathbf{w}_{2} & 
\mathbf{W}^{xy}  &= 2\mathbf{w}_{1}^{*} \wedge \mathbf{w}_{0} \\ 
\mathbf{W}^{xz}  &= 2\mathbf{w}_{2}^{*} \wedge \mathbf{w}_{0} & 
\mathbf{W}^{yz}  &= 2\mathbf{w}_{2}^{*} \wedge \mathbf{w}_{1} \\ 
\mathbf{W}^{x}   &= 2\mathbf{w}_{3}^{*} \wedge \mathbf{w}_{0} & 
\mathbf{W}^{y}   &= 2\mathbf{w}_{3}^{*} \wedge \mathbf{w}_{1} \\ 
\mathbf{W}^{z}   &= 2\mathbf{w}_{3}^{*} \wedge \mathbf{w}_{2} &
\mathbf{W}^{1}   &=  \mathbf{W}_{3}^{*} \wedge \mathbf{w}_{3} \\ 
\end{array}
\end{equation}
Again, the following properties hold:
\begin{equation}
\begin{array}{l@{}l@{}l@{}l@{}l@{}l@{}}
\mathbf{W}^{x^2} \cdot \mathbf{W}_{x^2} &= 1,~
\mathbf{W}^{y^2} \cdot \mathbf{W}_{y^2} &= 1,~
\mathbf{W}^{z^2} \cdot \mathbf{W}_{z^2} &= 1,~
\mathbf{W}^{xy} \cdot \mathbf{W}_{xy}   &= 1,~
\mathbf{W}^{xz} \cdot \mathbf{W}_{xz}   &= 1,\\
\mathbf{W}^{yz} \cdot \mathbf{W}_{yz}   &= 1,~
\mathbf{W}^{x} \cdot \mathbf{W}_{x}     &= 1,~
\mathbf{W}^{y} \cdot \mathbf{W}_{y}     &= 1,~
\mathbf{W}^{z} \cdot \mathbf{W}_{z}     &= 1,~
\mathbf{W}^{1} \cdot \mathbf{W}_{1}     &= 1\\
\end{array}
\end{equation}
Then, given $\mathbf{q}_{DPGA}$ the entity representing a quadratic of DPGA, any coefficients of this quadratic surface~\eqref{eq:genralQuadrics} can be extracted as:
\begin{equation}
\begin{array}{l@{}l@{}l@{}l@{}l@{}l@{}}
\mathbf{W}^{x^2} \cdot \mathbf{q}_{DPGA} &= \mathrm{a},~
\mathbf{W}^{y^2} \cdot \mathbf{q}_{DPGA} &= \mathrm{b},~
\mathbf{W}^{z^2} \cdot \mathbf{q}_{DPGA} &= \mathrm{c},~
\mathbf{W}^{xy} \cdot \mathbf{q}_{DPGA}   &= \mathrm{d},\\
\mathbf{W}^{xz} \cdot \mathbf{q}_{DPGA}   &= \mathrm{e},~
\mathbf{W}^{yz} \cdot \mathbf{q}_{DPGA}  &= \mathrm{f},~
\mathbf{W}^{x} \cdot \mathbf{q}_{DPGA}    &= \mathrm{g},~
\mathbf{W}^{y} \cdot \mathbf{q}_{DPGA}     &= \mathrm{h},\\
\mathbf{W}^{z} \cdot \mathbf{q}_{DPGA}     &= \mathrm{i},~
\mathbf{W}^{1} \cdot \mathbf{q}_{DPGA}    &= \mathrm{j}\\
\end{array}
\end{equation}
In a similar way as in projective geometry, the construction of a finite point of DPGA requires to add a homogeneous component 1 to the Euclidean components. The normalization of a point merely consists in dividing all the components by its $\mathbf{w}_{3}$ components(or $\mathbf{w}_{3}^*$ for the dual form) if it is a non-zero component.

\section{QCGA reciprocal operators}
For QCGA, quadratic surfaces can be either representing using the primal form or the dual form. We define the reciprocal operators for the dual form. Indeed, if one consider the primal form, then this would consist in computing the dual of the primal and apply the following reciprocal operators:
\begin{equation}
\begin{array}{c@{}lc@{}l}
\mathbf{Q}^{x^2} &=  \frac{1}{2} \mathbf{e}_{\infty1}& 
\mathbf{Q}^{y^2} &=  \frac{1}{2} \mathbf{e}_{\infty2} \\
\mathbf{Q}^{z^2} &=  \frac{1}{2} \mathbf{e}_{\infty3} & 
\mathbf{Q}^{xy}  &= \mathbf{e}_{\infty4} \\ 
\mathbf{Q}^{xz}  &= \mathbf{e}_{\infty5} & 
\mathbf{Q}^{yz}  &= \mathbf{e}_{\infty6} \\ 
\mathbf{Q}^{x}   &= \mathbf{e}_{1} & 
\mathbf{Q}^{y}   &= \mathbf{e}_{2} \\ 
\mathbf{Q}^{z}   &= \mathbf{e}_{3} &
\mathbf{Q}^{1}   &= \mathbf{e}_{o1}+\mathbf{e}_{o2}+\mathbf{e}_{o3} \\ 
\end{array}
\end{equation}
Given a general quadratic $\mathbf{q}^*$ whose coefficients are $(\mathrm{a},\mathrm{b},\mathrm{c},\cdots,\mathrm{j})$, the properties of these operators are as follows:
\begin{equation}
\begin{array}{l@{}l@{}l@{}l@{}l@{}l@{}}
\mathbf{Q}^{x^2} \cdot \mathbf{q}^{*} &= \mathrm{a},~
\mathbf{Q}^{y^2} \cdot \mathbf{q}^{*} &= \mathrm{b},~
\mathbf{Q}^{z^2} \cdot \mathbf{q}^{*} &= \mathrm{c},~
\mathbf{Q}^{xy} \cdot \mathbf{q}^{*}   &= \mathrm{d},~
\mathbf{Q}^{xz} \cdot \mathbf{q}^{*}   &= \mathrm{e},\\
\mathbf{Q}^{yz} \cdot \mathbf{q}^{*}   &= \mathrm{f},~
\mathbf{Q}^{x} \cdot \mathbf{q}^{*}    &= \mathrm{g},~
\mathbf{Q}^{y} \cdot \mathbf{q}^{*}     &= \mathrm{h},~
\mathbf{Q}^{z} \cdot \mathbf{q}^{*}     &= \mathrm{i},~
\mathbf{Q}^{1} \cdot \mathbf{q}^{*}     &= \mathrm{j}\\
\end{array}
\end{equation}
The construction of a QCGA point was already explained in the previous section. 
The reciprocal operation requires the computation of the normalization point $\mathbf{\hat{x}}$ of QCGA.
\begin{prop}
For QCGA point $\mathbf{x}$, the normalization is merely computed through an averaging of $\mathbf{e}_{o1},\mathbf{e}_{o2},\mathbf{e}_{o3}$ components thus of $\mathbf{e}_{o}$ component, namely as:   
\begin{equation}
-\frac{\mathbf{x}}{\mathbf{x} \cdot \mathbf{e}_{\infty}} 
\end{equation}
\end{prop}
\begin{proof}
A scale $\alpha$ on $\mathbf{x}$ acts the same way on all null basis vectors of $\mathbf{x}$:  
\begin{equation}
\begin{array}{r}
\alpha \mathbf{x} = \alpha\mathbf{x}_\epsilon  
+ \tfrac{1}{2} \alpha (x^2 \mathbf{e}_{\infty 1} +  y^2 \mathrm{w} \mathbf{e}_{\infty 2} + z^2 \mathbf{e}_{\infty 3}) 
 + xy \alpha \mathbf{e}_{\infty 4} + xz \alpha \mathbf{e}_{\infty 5} + yz \alpha \mathbf{e}_{\infty 6} \\
  + \alpha \mathbf{e}_{o1}+ \alpha \mathbf{e}_{o2} + \alpha \mathbf{e}_{o3}
\end{array}
\end{equation}
The metric of QCGA indicates (see Table~\ref{table:metric}):
\begin{equation}
\begin{array}{cl}
\alpha \mathbf{x} \cdot \mathbf{e}_{\infty 1} &= -\alpha \\ 
\alpha \mathbf{x} \cdot \mathbf{e}_{\infty 2} &= -\alpha \\ 
\alpha \mathbf{x} \cdot \mathbf{e}_{\infty 3} &= -\alpha \\ 
\end{array}
\end{equation}
Thus, if $\alpha\neq 0$ : 
\begin{equation}
\begin{array}{cl}
 \displaystyle \frac{-3 \alpha \mathbf{x}}{\alpha \mathbf{x} \cdot (\mathbf{e}_{\infty 1}+\mathbf{e}_{\infty 2}+\mathbf{e}_{\infty 3}) } \cdot \mathbf{e}_{\infty 1} &= - \frac{-3 \alpha \mathbf{x}}{-3 \alpha } \cdot \mathbf{e}_{\infty 1} \\ 
&= \mathbf{x} \cdot \mathbf{e}_{\infty 1} = -1
\end{array}
\end{equation}
A similar result is obtained with $\mathbf{e}_{\infty2}$ and $\mathbf{e}_{\infty3}$:
\begin{equation}
\begin{array}{cl}
\displaystyle  \frac{-3 \alpha \mathbf{x}}{\alpha\mathbf{x} \cdot (\mathbf{e}_{\infty 1}+\mathbf{e}_{\infty 2}+\mathbf{e}_{\infty 3}) } \cdot \mathbf{e}_{\infty 2} &= \mathbf{x} \cdot \mathbf{e}_{\infty 2} = -1
\end{array}
\end{equation}
\begin{equation}
\begin{array}{cl}
\displaystyle  \frac{-3 \alpha \mathbf{x}}{\alpha\mathbf{x} \cdot (\mathbf{e}_{\infty 1}+\mathbf{e}_{\infty 2}+\mathbf{e}_{\infty 3}) } \cdot \mathbf{e}_{\infty 3} &= \mathbf{x} \cdot \mathbf{e}_{\infty 3} = -1
\end{array}
\end{equation}
\end{proof}
Thus, we checked that for any scaled points $\mathbf{x}_1,\mathbf{x}_2$:
\begin{equation}
\displaystyle \frac{\mathbf{x}_1}{\mathbf{x}_1 \cdot \mathbf{e}_{\infty}} \cdot \frac{\mathbf{x}_2}{\mathbf{x}_2 \cdot \mathbf{e}_{\infty}} = - \displaystyle \frac{1}{2} \left\| \mathbf{x}_{1\epsilon}-\mathbf{x}_{2\epsilon} \right\|^2 
\end{equation} 
The extraction of the Euclidean components $(x,y,z)$ of a normalized point $\mathbf{\hat{x}}$ of DCGA can be performed as follows:
\begin{equation}
x = \mathbf{\hat{x}} \cdot \mathbf{e}_1, ~~~~~ 
y = \mathbf{\hat{x}} \cdot \mathbf{e}_2, ~~~~~
z = \mathbf{\hat{x}} \cdot \mathbf{e}_3
\end{equation}

\section{Test}
We tested this approach by defining an ellipsoid from 9 points using QCGA. Then we rotate it using DPGA and back-convert the rotated ellipsoid into QCGA framework. In terms of Geometric Algebra computations, first we compute the quadratic:
\begin{equation}
\mathbf{q}^* = (\mathbf{x}_1 \wedge \mathbf{x}_2 \wedge  \cdots \wedge \mathbf{x}_9 \wedge \mathbf{I}_{o}^{ \rhd } )^*
\end{equation}
Then, we apply the extraction operators of QCGA to convert the QCGA quadratic to DPGA quadratic.
\begin{equation}
\begin{array}{c@{}l}
\mathbf{q}_{DPGA} =& 4 (\mathbf{Q}^{x^2}\cdot \mathbf{q}^*) \mathbf{w}_0^{*} \wedge \mathbf{w}_0 + 4(\mathbf{Q}^{y^2}\cdot \mathbf{q}^*) \mathbf{w}_1^{*} \wedge \mathbf{w}_1 + 4(\mathbf{Q}^{z^2}\cdot \mathbf{q}^*)  \mathbf{w}_2^{*} \wedge \mathbf{w}_2 \\
& + 4 (\mathbf{Q}^{1}\cdot \mathbf{q}^*) \mathbf{w}_3^{*} \wedge \mathbf{w}_3 \\
&+ 2 (\mathbf{Q}^{xy}\cdot \mathbf{q}^*) ( \mathbf{w}_0^{*} \wedge \mathbf{w}_1+ \mathbf{w}_1^{*} \wedge \mathbf{w}_0) \\
&+ 2(\mathbf{Q}^{xz}\cdot \mathbf{q}^*) (\mathbf{w}_0^{*} \wedge \mathbf{w}_2 + \mathbf{w}_2^{*} \wedge \mathbf{w}_0)\\
& + 2 (\mathbf{Q}^{yz}\cdot \mathbf{q}^*) ( \mathbf{w}_1^{*} \wedge \mathbf{w}_2 + \mathbf{w}_2^{*} \wedge \mathbf{w}_1) \\
&+ 2(\mathbf{Q}^{x}\cdot \mathbf{q}^*) (\mathbf{w}_0^{*} \wedge \mathbf{w}_3+ \mathbf{w}_3^{*} \wedge \mathbf{w}_0) \\ 
&+ 2(\mathbf{Q}^{y}\cdot \mathbf{q}^*) (\mathbf{w}_1^{*} \wedge \mathbf{w}_3 + \mathbf{w}_3^{*} \wedge \mathbf{w}_1) \\
&+ 2(\mathbf{Q}^{z}\cdot \mathbf{q}^*)(\mathbf{w}_2^{*} \wedge \mathbf{w}_3 + \mathbf{w}_3^{*} \wedge \mathbf{w}_2) 
\end{array}
\end{equation}
The rotation is now performed as follows:
\begin{equation}
\mathbf{q}_{DPGA} = \mathbf{R} \mathbf{q}_{DPGA} \mathbf{R}^{-1} 
\end{equation}
The rotor $\mathbf{R}$ is defined as:
\begin{equation}
\mathbf{R} = \exp( \frac{1}{2} \theta \mathbf{w}_{i}\mathbf{w}_{j}^*)
\end{equation}
where $i \neq j$ 

The final step is to back-convert the resulting quadric to the QCGA framework. It is merely computed using the QCGA extraction operators as follows:
\begin{align}
  \mathbf{q}^* &= -\big(2 (\mathbf{W}^{x^2} \cdot \mathbf{q}_{DPGA} ) \mathbf{e}_{o1}
         +2 (\mathbf{W}^{y^2} \cdot \mathbf{q}_{DPGA} )  \mathbf{e}_{o2}
         +2 (\mathbf{W}^{z^2} \cdot \mathbf{q}_{DPGA} )  \mathbf{e}_{o3}
         \nonumber \\
  &\phantom{==}
         + (\mathbf{W}^{xy} \cdot \mathbf{q}_{DPGA} )  \mathbf{e}_{o4}  		+(\mathbf{W}^{xz} \cdot \mathbf{q}_{DPGA} )  \mathbf{e}_{o5}
         + (\mathbf{W}^{yz} \cdot \mathbf{q}_{DPGA} )  \mathbf{e}_{o6}\big)
  \nonumber \\
  &\phantom{==}
         +\big( (\mathbf{W}^{x} \cdot \mathbf{q}_{DPGA} ) \mathbf{e}_{1} + (\mathbf{W}^{y} \cdot \mathbf{q}_{DPGA} ) \mathbf{e}_{2} + (\mathbf{W}^{z} \cdot \mathbf{q}_{DPGA} ) \mathbf{e}_{3} \big)  \nonumber \\
  &\phantom{==}
         -\frac{(\mathbf{W}^{1} \cdot \mathbf{q}_{DPGA} )}{3}(\mathbf{e}_{\infty 1}+\mathbf{e}_{\infty 2}+\mathbf{e}_{\infty 3}).
  \label{eq:dualq}      
\end{align}
Note that the program can be found in the plugin folder of the git repository previously shown.

\section{Conclusion}
 In this paper, we focused on a hybrid approach to deal with quadratic surfaces. After presenting the main approaches to represent and manipulate quadratic surfaces, we introduced a new hybrid Geometric Algebra approach. This approach unifies all the models of Geometric Algebra into one more general approach that allows to represent any quadratic surface either using control points or from the coefficients of its implicit form. We showed that the proposed method also enables to easily extract geometric properties, like the curvature, in an efficient way. For the following, we seek for a generalisation of this approach for the representation of quadratic and cubic surfaces. Some frameworks are considered and all of these models require high dimensional frameworks.

\bibliographystyle{acm}
\bibliography{GAbiblio} 

\end{document}